\documentclass[11pt,english,letterpaper]{article}
\usepackage[top=1in,bottom=1in,left=1in,right=1in]{geometry}
\usepackage[utf8]{inputenc}

\usepackage{amsmath,amssymb,amsthm,amsfonts}
\usepackage{graphicx}
\usepackage{pgf, tikz}
\usetikzlibrary{arrows}
\usepackage{xspace, units}
\usepackage{algorithm}
\usepackage[noend]{algpseudocode}
\usepackage{mathdots}
\usepackage{multicol}
\usepackage{tabu,multirow,boldline}
\usepackage{cite,nicefrac}
\usepackage{xcolor}
\usepackage[breaklinks]{hyperref}
\usepackage[ngerman,english]{babel}
\usepackage{cleveref}
\usepackage[breaklinks]{hyperref}
\usepackage[ngerman,english]{babel}

\allowdisplaybreaks

\newcommand{\ubox}[1]{\fbox{\parbox[c][1cm][c]{3cm}{\centering #1}}}
\definecolor{Darkblue}{rgb}{0,0,0.4}
\definecolor{Brown}{cmyk}{0,0.81,1.,0.60}
\definecolor{Purple}{cmyk}{0.45,0.86,0,0}
\hypersetup{colorlinks=true,
            citebordercolor={.6 .6 .6},linkbordercolor={.6 .6 .6},
            citecolor=black,urlcolor=Darkblue,linkcolor=black,
            pdfpagemode=UseThumbs,pdfstartview=FitH}

\newcommand{\lref}[2][]{\hyperref[#2]{#1~\ref*{#2}}}

\newtheorem{theorem}{Theorem} 
\newtheorem{lemma}{Lemma} 
\newtheorem{Definition}{Definition}

\newcommand{\NN}{\ensuremath{\mathbb{N}}}

\newcommand{\Pro}[1]{\mbox{\rm\bf Pr}\left[#1\right]}

\newcommand{\classP}{{\sf P}}

\usepackage{xspace}
\newcommand{\squote}[1]{`#1'\xspace}
\newcommand{\dquote}[1]{``#1''\xspace}
\makeatletter%
\if@twocolumn%
  \newcommand{\condtwocol}[1]{#1}
\else
  \newcommand{\condtwocol}[1]{}
\fi
\makeatother
\author{Rebecca Reiffenh\"auser\thanks{Sapienza University of Rome -- rebeccar@diag.uniroma1.it }}
\title{An Optimal Truthful Mechanism\\ for the Online Weighted Bipartite Matching Problem\thanks{I dedicate this work to my advisor Prof. Berthold Vöcking, who posed the problem to me and tragically died in 2014. He was convinced a constant-competitive truthful mechanism must exist.}~\thanks{Partially supported by the ERC Advanced Grant 788893 AMDROMA \dquote{Algorithmic and Mechanism Design Research in Online MArkets}.}}
\date{}
\begin{document}
\selectlanguage{english}
\maketitle

\begin{abstract}
In the weighted bipartite matching problem, the goal is to find a maximum-weight matching in a bipartite graph with nonnegative edge weights. We consider its online version where the first vertex set is known beforehand, but vertices of the second set appear one after another. Vertices of the first set are interpreted as items, and those of the second set as bidders. On arrival, each bidder vertex reveals the weights of all adjacent edges and the algorithm has to decide which of those to add to the matching. 
We introduce an optimal, $e$-competitive truthful mechanism under the assumption that bidders arrive in random order (secretary model). 

It has been shown that the upper and lower bound of $e$ for the original secretary problem extends to various other problems even with rich combinatorial structure, one of them being weighted bipartite matching.
But truthful mechanisms so far fall short of reasonable competitive ratios once respective algorithms deviate from the original, simple threshold form. 
The best known mechanism for weighted bipartite matching by Krysta and V\"ocking~\cite{KrystaV12} offers only a ratio logarithmic in the number of online vertices. We close this gap, showing that truthfulness does not impose any additional bounds. 
The proof technique is new in this surrounding, and based on the observation of an independency inherent to the mechanism. The insights provided hereby are interesting in their own right and appear to offer promising tools for other problems, with or without truthfulness.

\end{abstract}

\section{Introduction}
We consider an auction setting where, while the set of items is known, bidders arrive in an online fashion. 
Assignments are restricted to a one-on-one manner where every bidder can be assigned at most one item, and each item can be assigned at most once. 
Each bidder has nonnegative valuations for each of the items as private information, and reveals those on arrival.
The algorithm then has to make a definitive decision which item to assign the arriving bidder, and cannot revoke its choice later. 
The goal is to find a one-on-one assignment, or \emph{matching} of items to bidders such that the social welfare, i.e. the sum of all valuations in the assignment, is maximized.

We assume the number of bidders to be known beforehand, and their order of arrival to be uniformly at random. This assumption is not only reasonable since real-world bidders are unlikely to conspire about when to place their bids. It is also necessary in order to give any guarantee on the quality of our outcome: When both valuations and the arrival order are arbitrary, an adversary can always render any online algorithm useless by choosing the very last bidder's valuations to be extremely high or low (see~\cite{AggarwalGKM11}).
The emphasis in this paper lies on the fact that valuations are the private information of the bidders, who act as selfish agents and might misreport them strategically. This makes the problem one of mechanism design, and we strive to make truthtelling every bidder's best interest by charging according \emph{prices} accompanying the assignments themselves.
The underlying weighted bipartite matching problem is a classic setting from graph theory, and one of the fundamental problems in operations research where it is called simply the \emph{assignment problem}. Applications for its game-theoretic online version with selfish bidders are vast and lie mainly in e-commerce and online auctions, e.g. for advertising.

The online problem considered is an extension of the well-known \emph{secretary problem}, where a known number of applicants for a position arrive one after another and the goal is to give the position to the best candidate. The best online algorithm for the secretary problem has an expected competitive ratio\footnote{The competitive ratio denotes the (expected) ratio between an online algorithm's outcome and the value of an offline optimum.} of $e$, which is also the lower bound (see, e.g.,~\cite{Dynkin63}).
While this guarantee has been shown to transfer to other online problems even with rich combinatorial structure, like weighted bipartite matching or combinatorial auctions with submodular valuations~\cite{KesselheimRTV13}, the situation is less clear for the strategic setting.
Here, truthful mechanisms can be derived easily for the original secretary problem and variations where algorithms still follow the same threshold principle. But when other algorithmic approaches are required, guarantees for truthful mechanisms are nowhere near the lower bound of $e$. For weighted bipartite matching, this gap is as large as an upper bound of $O(\log n)$~\cite{KrystaV12}, i.e. logarithmic in the number of bidders.

The difficulty one has to overcome in closing this gap lies less in finding a promising mechanism, and more in identifying appropriate ways to analyze its performance. We introduce a new technique for this, which is based on the observation of a simple, yet powerful independency property inherent to the mechanism. For online weighted bipartite matching, this shows that indeed, truthfulness poses no additional restriction on the competitive ratio. For other, related problems, we believe that (a weaker form of) the same property could serve as a useful guideline to finding near-optimal algorithms with or without truthfulness.

\section{Related Work}

In the adversarial arrival model, the first to research online bipartite matching were Karp et al.~\cite{KarpVV90}. They introduced a randomized algorithm that has expected competitive ratio $e/(e-1)$ (for a simple proof, see~\cite{BirnbaumM08}) and also showed a matching lower bound. This lower bound, however, does not hold in the random arrival model, which was shown by Karande et al.~\cite{KarandeMT11} and Mahdian and Yan \cite{MahdianY11}.
Online matching with random arrival order was first studied in ~\cite{GoelM08}, who analyzed a greedy approach in a budgeted setting. 

For (restricted) weighted versions of the problem, Aggarwal et al. \cite{AggarwalGKM11} first investigated online matching with adversarial arrival order. In the case that all edges incident to the same offline vertex have the same weight, they obtained an expected competitive ratio of $e/(e-1)$. Under the assumption that the edge weights represent a metric space, a 3-competitive algorithm was given by Kalyanasundaram and Pruhs ~\cite{KalyanasundaramP93}.

Our problem, online weighted bipartite matching with random arrival order, is a generalization of the matroid secretary problem on transversal matroids\footnote{This problem can be viewed as the special case of weighted bipartite matching where all edges incident to the same online vertex have the same weight, or alternatively, weights are on online vertices instead of edges.}.
This closely related problem was introduced in~\cite{BabaioffIK07} by Babaioff et al., who presented an approach that is constant-competitive for bounded left degree, later improved upon by Dimitrov and Plaxton\cite{DimitrovP12}. 
These results are part of a line of work that is central to the field, and evolves around the famous conjecture stated by Babaioff, Immorlica and Kleinberg~\cite{BabaioffIK07} in 2007, saying that a $O(1)$-approximation exists for the secretary problem on general matroids.
The online weighted bipartite matching problem itself was first tackled in a 2009 paper by Korula and P\'{a}l~\cite{KorulaP09}, who gave an 8-competitive algorithm based on a combination of sampling and a greedy approach. Kesselheim et al.~\cite{KesselheimRTV13} introduced an $e$-competitive algorithm later, matching the secretary lower bound. 
They also presented an extension of this result to online combinatorial auctions with submodular valuation functions.

The term \emph{online mechanism design} for strategic settings was initially coined by Friedman and Parkes~\cite{FriedmanP03} in 2003. 
In later publications, Parkes et al. analyzed the relation of such online mechanisms to Markov decision processes~\cite{ParkesS04}, and gave alternative approximative strategies for the case that the MDP's size becomes impractical\cite{ParkesSY05}. They, as much of the work in the field, assume a sequential decision model where agents have arrival (and departure) times, and their valuations are sampled from a probability distribution. Online mechanism design has since been combined with the assumption of randomized arrivals and secretary problems. For example, Hajiaghayi et al.~\cite{HajiaghayiKP04} and in the following, Kleinberg~\cite{Kleinberg05} considered the case of auctioning $k$ identical goods to agents, while Gershkov and Moldovanu~\cite{GershkovM10} studied the case of heterogeneous items in a deadline model. Work in this area often focuses on online models quite different from the one in our work, featuring discrete time periods or expiring goods (see, e.g.~\cite{GerdingJPRRS13},~\cite{LaviN05}).
The best known online mechanism for weighted bipartite matching in the secretary model was given by Krysta and Vöcking \cite{KrystaV12} in 2012, who presented online mechanisms for several combinatorial auction problems. Their truthful mechanism for weighted bipartite matching has competitive ratio logarithmic in the number of bidders.

\section{Notation and Preliminaries}

\paragraph{Problem Definition}
In the weighted bipartite matching problem, we are given a graph $G= (V, E)$ with $V= I\cup J$ and $E\subseteq I\times J$. All edges $e\in E$ come with a nonnegative weight $w(e)$, and the goal is to find a maximum-weight matching $M$.
A matching $M$ in $G$ is a set $M\subseteq E$ such that for all $e, e'\in M$, $e=(i, j)$ and $e' = (i', j')$, we have $i\neq i'$, $j\neq j'$.
The offline problem can be solved in polynomial time, for example via the famous \emph{Hungarian method} introduced in \cite{Kuhn55}. 
We define the according \emph{online} problem as follows:
While the vertex set $J$ and size $n=|I|$ of vertex set $I$ are known from the beginning, the vertices $i\in I$ themselves arrive one by one according to an order $\pi$, which is chosen uniformly at random.
On arrival, each $i\in I$ reveals the weights $w(i,j)$, $j\in J$ of all incident edges and the algorithm has to decide irrevocably which of those to add to the matching (or not to pick any).\\
From now on, we denote a maximum-weight matching in the induced subgraph on vertex sets $I'\subseteq I$ and $J'\subseteq J$ as $M_{OPT}(I',J')$, and the sum of all edge weights in such a matching as $OPT(I',J')$.

\paragraph{Strategic Setting}
We assume the vertices in $I$ to be the bidders, and those in $J$ to be the items.
The weight of some edge $e=(i,j)$ then corresponds to the private valuation bidder $i\in I$ has for an item $j\in J$. 
We strive to maximize the sum over all valuations of edges in our matching, i.e. the social welfare $\sum_{e\in M} w(e)$.
Bidders act as selfish agents and the valuations they have on the items (i.e., the edge weights) are their private information. In this strategic setting, we want to find a mechanism, i.e. an algorithm together with a pricing scheme for the items, with a specific property: A bidder might choose to misreport the edge weights when it helps him maximize his own gain or \emph{utility}, defined as the difference between the bidder's valuation of the item assigned to him and the price he has to pay for it.
To make sure this does not happen, we want to provide a \emph{truthful} mechanism. A mechanism is \emph{truthful} if reporting their true valuations is a dominant strategy for the bidders, i.e. misreporting a bidder's valuations never increases his utility as long as the rest of the problem instance remains the same.

\paragraph{Offline Mechanism}
Note that since the problem is in $\classP$, one can efficiently compute an offline optimum and use the famous Vickrey-Clarke-Groves mechanism \cite{Vickrey61, Clarke71, Groves73}. The principle here is as follows: For an optimal assignment $M_{OPT}$ and each item $j\in J$, charge the winning bidder $i$ with $(i,j)\in M_{OPT}$ a price $p(i,j)$ according to his \emph{externality}. The externality for bidder $i$ and item $j$ is defined as the loss in social welfare for the other bidders, caused by assigning $j$ to $i$. Formally,
\[ p(i,j)= OPT(I\setminus\{i\}, J)-OPT(I\setminus\{i\}, J\setminus\{j\}) \enspace.\] 

This ingenious truthful mechanism is based on the observation that for above prices, the interests of bidders and mechanism designer almost magically align: Bidders strive to maximize their utility $w(i,j)-p(i,j)$, depicting how much more they value an item compared to what they pay for it. The designer is interested in maximizing the weight of a possible edge incident to $i$, together with the weight of an optimal matching in the rest of the graph -- which also amounts to assigning $i$ a utility-maximizing item under above prices.
Our mechanism for the online problem will make ample use of such VCG pricings.

\paragraph{Online Mechanisms}

In the online setting, algorithmic difficulty arises from the obvious lack of information on future edge weights together with the irrevocability of the mechanism's decisions. For the non-strategic online problem, the algorithm of Kesselheim et al.\cite{KesselheimRTV13} handles this difficulty optimally, which is clearly a very good point for us to start. 
We try to combine a variation of their strategy with the requirement of truthfulness, and derive a mechanism that is applicable for the case of selfish bidders.
Truthful mechanisms with an optimal competitive ratio of $e$ do already exist for the original secretary problem, as well as a number of variations such as multiple-choice secretary. The reason is simple: For all of these problems, according algorithms follow the same principle. First, they exploit the randomness in the arrival order and observe a constant fraction of the problem instance, in expectation containing a good \emph{sample} of the values to come. Then, they use the gathered information to determine a \emph{threshold} and choose the next arriving candidate if and only if it beats that value.
This type of algorithm can be turned into a truthful mechanism by simply charging the threshold values as prices.

For weighted bipartite matching (and other problems), bidders' valuations cannot be represented by a single value and feasible sets of assignments/edges are nontrivial in structure. Therefore, algorithms of above type are quite obviously insufficient. 
Still, ensuring truthfulness is easy in principle, and can be done simply by determining fixed item prices \emph{before} an arriving bidder reveals his valuations. Then, as long as the bidder is assigned a utility-maximizing item, he will have no incentive to lie\footnote{Note that both~\cite{KorulaP09} and~\cite{KesselheimRTV13} fail to be truthful at precisely this point: Informally, the former algorithm relies on assigning an item with highest \emph{valuation}, the latter assigns no item at all if the \squote{correct} one is no longer available.}.
The real question is how to still guarantee a good competitive ratio (and it was so far not clear whether this is possible for reasonable, e.g. constant, ratios).

\section{Our Results}
We provide an optimal $e$-competitive, truthful mechanism for the online weighted bipartite matching problem as stated above, closing the gap between the known lower bound of $e$ and upper bound of $O(\log n)$~\cite{KrystaV12}. 
Bipartite weighted matching covers as special cases a number of other problems. In consequence, our mechanism is also applicable to the multiple-choice secretary problem, and the secretary problem on transversal matroids.
\begin{theorem}\label{mainthm} There exists a truthful mechanism for the online weighted bipartite matching problem with random arrival order that has expected competitive ratio of $e$.
\end{theorem}
On the algorithmic side, the mechanism is fairly simple and bases its assignment decisions on the computation of local optima in the known part of the graph. 
The competitive analysis, however, is less straightforward. We build it on the observation of an independency inherent to the mechanism, which together with the assignment rule yields the optimal ratio of $e$. The independency seems surprising at first, but in hindsight, this property of the mechanism is not a lucky coincidence. Much more, it represents very reasonable design guidelines that could prove useful for other problems in the same scope.
More exactly, we prove that in expectation over possible arrival orders $\pi$, the development of the set of yet unassigned items does not depend on that of the set of known bidders. It holds:

\begin{lemma}\label{informalindep}
Under the condition that $t$ bidders have arrived and $s$ items are unassigned, the sets of arrived bidders and unassigned items are independent random variables.
\end{lemma}
As a consequence, the probability that any single item $j\in J$ is still available can be stated only with regard to the current step $t$ of the mechanism, making it considerably easier to track.

In contrast to previous work in this surrounding, our mechanism handles bidders with complex private information (as opposed to single values)  while the feasible edge sets are also described by a non-trivial combinatorial structure. To our knowledge, it is the first constant-competitive mechanism for such an online problem beyond the aforementioned threshold-price mechanisms, and models with considerably restricted valuations. Still, truthfulness here does not lead to any deterioration in the competitive guarantee compared to the optimal online algorithm and lower bound. 

\section{The Mechanism}
The mechanism we propose follows simple guidelines common also to other secretary-type problems. First, we wait for a fixed number $k$ of steps and only observe the reported edge weights without making any assignments. This is called the sampling phase. With the arrival order being uniformly at random and the total number $n$ of bidders known beforehand, the first $k$ steps give, in expectation, a good overview of what kind of valuations to expect: When $k$ is some constant fraction of $n$, e.g. the chance of seeing the largest overall bid during this sampling phase is also constant.
Second, we choose to add edges to our matching according to an offline optimum over the set of arrived bidders and yet unassigned items. This ensures, together with the sampling phase, that the expected weight of any added edge be considerably large with respect to the overall offline optimum. All in all, we get 
Algorithm\nobreakspace \ref{TheMechanism1}
(disregarding the strategic component for now).
\begin{algorithm}
\caption{}
\label{TheMechanism1}
\begin{itemize}
\item Initially, set $M=\emptyset$, $I_0=\emptyset$, and $J_U=J$.
\item For steps $t= 1,\dots, k$: \\ Observe valuations $w(i_t, j)$ for all $j\in J$, reported by the arriving bidder $i_t$. \\ Set $I_t=I_{t-1}\cup\{i_t\}$.
\item For steps $t=k+1,\dots, n$: 
	\begin{itemize}
	\item[] Observe valuations $w(i_t,j)$ for all $j\in J$ reported by the arriving bidder $i_t$.
	\item[] Set $I_t= I_{t-1}\cup\{i_t\}$.
	\item[] \textbf{If there exists $j^*\in J_U$ such that $(i_t, j^*)\in M_{OPT}(I_t, J_U)$:}
	\begin{itemize}
	\item[] \textbf{assign $i_t$ the item $j^*$}
	\item[] and set $M=M\cup \{(i_t, j^*)\}$, $J_U= J_U\setminus \{j^*\}$.
	\end{itemize}
	\end{itemize}
\end{itemize}
\end{algorithm}

So after the sampling phase, the algorithm simply assigns each arriving bidder $i_t$ the item (if any) that would have been assigned to him in the optimal matching over all known bidders and yet unassigned items. From here on, let us denote the partial optimum $M_{OPT}(I_t,J_U)$ determining the possible assignment in step $t$ of the mechanism simply as $M_t$.
To turn this into a truthful mechanism, we employ a payment scheme using the aforementioned VCG prices. As long as each partial offline optimum computed by the algorithm is unique (and we make sure it is by tie-breaking), the mechanism 
(Algorithm\nobreakspace \ref {TheMechanism2})
makes exactly the same assignments as 
Algorithm\nobreakspace \ref {TheMechanism1}.

\begin{algorithm}
\caption{}
\label{TheMechanism2}
\begin{itemize}
\item Initially, set $M=\emptyset$, $I_0=\emptyset$, and $J_U=J$.
\item For steps $t= 1,\dots, k$: \\ Observe valuations $w(i_t, j)$ for all $j\in J$, reported by the arriving bidder $i_t$. \\ Set $I_t=I_{t-1}\cup\{i_t\}$.
\item For steps $t=k+1,\dots, n$: 
	\begin{itemize}
	\item[] Set $I_t= I_{t-1}\cup\{i_t\}$.
	\item[] \textbf{Set item prices $p_j(t)=OPT(I_{t-1}, J_U) - OPT(I_{t-1}, J_U\setminus \{j\})$ for all $j\in J_U$.}
	\item[] Observe valuations $w(i_t,j)$ for all $j\in J$, reported by the arriving bidder $i_t$.
	\item[] \textbf{If there exists $j\in J_U$ with $w(i_t,j)\geq p_j(t)$:}
	\begin{itemize}
	\item[] \textbf{assign $i_t$ the item $j^*\in J_U$ such that $w(i_t,j^*) - p_{j^*}(t)$ is maximized}
	\item[] and set $M=M\cup \{(i_t, j^*)\}$, $J_U= J_U\setminus \{j^*\}$. 
	\end{itemize}
	\end{itemize}
\end{itemize}
\end{algorithm}

The computation of each partial optimum $M_t$ now is implicit, and replaced by the computation of the prices, which are determined via the \emph{value} of a partial optimum without the current bidder (and how much it deteriorates when any item $j$ is taken out). 
Note that these are exactly the VCG prices for $M_t=M_{OPT}(I_t, J_U)$. Therefore, the existence of an item with nonnegative utility for the arriving bidder $i_t$ under these prices is equivalent to $i_t$ receiving some item in $M_t$. Moreover, the item which maximizes this utility is the unique item assigned to $i_t$ in $M_t$.

\paragraph{Feasibility and Truthfulness}
Since each item is removed from $J_U$ on assignment, and each bidder is considered only once on arrival, the resulting $M$ is clearly a matching: No vertex can ever be assigned more than once. 
The latter, mechanism formulation also ensures truthfulness. This is due to the simple fact that all prices are fixed before a bidder even reveals his valuations: Obviously, what he reports has no influence whatsoever on what he has to pay for any item. Therefore, the best he can hope for is to receive an item which maximizes his utility under those prices (if said utility is nonnegative). The mechanism's assignment rule does give him exactly this, i.e. he gets his favourite choice when reporting the truth. 

\vspace{5mm}
Note that our Algorithm\nobreakspace \ref {TheMechanism1} differs from the $e$-competitive algorithm of Kesselheim et al.\cite{KesselheimRTV13} only in the local optima computed: They consider an optimal matching on all known bidders and \emph{all} items, while we only take into account those items that have not yet been assigned.
Interestingly, on the one hand, their method appears incompatible with truthfulness since it denies a bidder any assignment if his \squote{correct} item is no longer available (thus incentivizing him to misreport his valuation for such items to be very low).
On the other hand, our method appears incompatible with their elegant technique of analyzing the algorithm's behavior in a backward fashion because early assignments affect later local optima in a seemingly uncontrollable way.
Consequently, we employ other, new means to show our result, which needs a little preparation.

\section{The Mechanism's Independency Property}
\subsection{Assumptions and Definitions}
For our competitive analysis, we start with proving a formalized version of Lemma \ref{informalindep}, for which we
make a number of assumptions:
\begin{itemize}
\item The graph $G$ is complete, i.e. the edge set $E$ equals $I\times J$. This is essentially the assumption of \emph{free disposal} common to mechanism design: Any bidder would be willing to accept any item for a price of zero, since he can just throw it away.
\item The length $k$ of the sampling phase is at least $m=|J|$. Together with the first assumption, this implies that in each computation of a partial optimum during the mechanism, the according matching $M_t$ actually assigns each item (i.e., $|M_t|=|J_U|$). Note this also implies that the number of bidders $n$ be sufficiently large. This fact, however, poses no restriction to the instances we can handle: For any instance $(I, J)$ with $n=|I|$ insufficiently small, we instead consider $(I\cup I', J)$, where $I'$ consists of $n'$ many \emph{dummy} bidders with valuation $0$ for all items. Now, for each $i'\in I'$, choose a number $t_{i'}$ from $\{1,\dots, n+n'\}$ uniformly at random (without replacement), and in each according step $t_{i'}$ of the mechanism, present it with the bidder $i'$. In the remaining $n$ \emph{unreserved} steps, the bidders in $I$ will be arriving in random order, as before. With this, $(I\cup I')$ as a whole is considered in random order and the mechanism runs exactly as it would have if $(I\cup I', J)$ had been an original problem instance.
Obviously, an optimal matching in $(I\cup I', J)$ has the same weight as one in $(I,J)$, so the expected competitive ratio is still with regard to the same optimal value.
\item We employ a consistent tie breaking that ensures each matching computed during the mechanism to be \emph{unique}. This can be done, for example, via an arbitrary but fixed order on the sets $I$ and $J$, implying that identifiers/names of the bidders have to be known beforehand (nonanonymous mechanism). Alternatively, one could also assume existence of a lexicographic order on the bidders (restricting oneself to instances where no two bidders have the exact same type). Other types of tie-breaking are in principle also possible. However, perturbation of the edge weights would result in only approximate truthfulness, and randomization of the chosen assignments in a need to considerably adjust our proofs. 
\end{itemize}
Our analysis heavily relies on statements about the mechanism's situation after a certain step $t$, for which we introduce the following notation.

\begin{Definition} We define the tuple $(I_t, J_U)$ with $I_t\subseteq I$, $J_U\subseteq J$ of the bidders arrived and the items yet unassigned after step $t$ of the mechanism as the mechanism's \emph{state} after step $t$.
\end{Definition}
Note that the state after step $t$ together with the arriving bidder in step $t+1$ defines the mechanism's state after step $t+1$: The mechanism will remove the item that is assigned to the last arriving bidder in $M_{t+1}=M_{OPT}(I_{t+1}, J_U)$ from the available set. As $M_{t+1}$ is the unique optimum in the given subgraph, the order in which all the vertices in $I_t$ arrived is irrelevant for the result.

\begin{Definition} We define $\mathcal{S}(t, s)$ as the set of all possible mechanism states with $t$ arrived bidders and $s$ available items:
\[ \mathcal{S}(t,s) = \left\{(I_t, J_s)\,\, |\,\, I_t\subseteq I, J_s\subseteq J, |I_t|=t, |J_s|=s\right\}\]
\end{Definition}\

\begin{Definition} Let $T$ be a table or matrix with $n$ rows and $n!$ columns. We say each row $T(t,\cdot)$ of $T$ depicts the $t$-th step of the mechanism, and each column $T(\cdot, l)$ a specific choice $\pi_l$ for the random order $\pi$. We define each entry $T(t,l)$ to contain the state $(I_t, J_U)$ of the mechanism after $t$ steps under the assumption that the random arrival order is $\pi_l$. 
\end{Definition}
This means, $T$ lists all possible (due to choice of $\pi$), different runs of the mechanism together with all states reached in the respective steps, see 
Table\nobreakspace \ref {table:StatesForAllRandomOrders}. 
As tracking probabilities for certain events along the arrival of bidders seems hard, especially with all the dependencies on already-assigned items etc., we will avoid doing so directly.
Instead, we will use $T$ to punctually analyze what happens during the mechanism.
\begin{table}
\begin{tabu} to \columnwidth {V{3}c|cV{3}cX[c]|c|X[c]cV{3}}
\hlineB{3}
\multicolumn{2}{V{3}cV{3}}{} & \multicolumn{5}{cV{3}}{random order}\\
\cline{3-7}
\multicolumn{2}{V{3}cV{3}}{} & $\pi_1$ & $\cdots$ & $\pi_l$ & $\cdots$ & $\pi_{n!}$\\
 \hlineB{3}
& $1$ & $\ddots$ & &$\vdots$ & & $\iddots$ \\
& \multirow{2}{*}{$\vdots$} & & & & & \\
& & & & & & \\
\cline{2-7}
step & $t$ & $\cdots$ & & $S(t,l)$ & & $\cdots$ \\
\cline{2-7}
& \multirow{2}{*}{$\vdots$} & & & & & \\
& & & & & & \\
& $n$ & $\iddots$ & & $\vdots$ & & $\ddots$\\
 \hlineB{3}
\end{tabu}
\caption{The algorithm's states for all random orders.}
\label{table:StatesForAllRandomOrders}
\end{table}

The connection between $T$ and the probability distribution over $\pi$ is straightforward: The likeliness of a situation after a certain step $t$ equals the fraction of entries in $T(t,\cdot)$ that depict it.
For example, if a certain state appears in a certain fraction of all $n$-th row entries, this directly corresponds to the chance the mechanism will end in that state.
We will use similar arguments as this one to get a grasp on the important developments throughout all possible mechanism runs.
Note again that $T(t,l)$ is well-defined for all $t\in \{1,...,n\}$ and all $l\in \{1,..., n!\}$ due to our use of tie breaking and the uniqueness of $M_t$. 

To gain some overview of what is happening during the mechanism, we will try to analyze how often each possible state $(I_t, J_U)$ appears in each row $T(t,\cdot)$ of $T$.

\begin{Definition}
For some state $S$ of the mechanism, we call the number of different tuples $(t,l)$ with $T(t,l)=S$ the multiplicity of state $S$ in $T$.  Shortly, we denote this number as $mul(S)$.
\end{Definition}

\begin{Definition} We define $\pi^t$ as the first $t$ bidders/choices of $\pi$.
\end{Definition}

\begin{Definition} If there exists a choice $\pi_l$ for $\pi$ such that state $S'$ of the mechanism is reached after step $t$, and state $S$ after step $t+1$ of the mechanism for arrival order $\pi_l$, we call $S'$ a predecessor of $S$. We define $P(S)$ as the set of all predecessors of $S$.
\end{Definition}

\subsection{The Independency Lemma}
With above preparations, we prove a concise version of our informal independency statement from above (Lemma \ref{informalindep}), which is central to our analysis. It states that in each row of $T$, each combination of exactly  $t$ bidders and $s$ items appears equally often.
To be more exact, we prove that for any fixed input graph, the number of different random orders $\pi$ for which the mechanism reaches a certain state $S\in \mathcal{S}(t,s)$ is  a function of only $t$ and $s$.
The proof is via a purely combinatorial analysis of the entries in $T$.

\begin{lemma}
\label{lem:independency}
For all  
$t\leq n$, $s\leq m$: Let $S, S' \in \mathcal{S}(t,s)$. Then, $mul(S) = mul (S')$.
\end{lemma}
\begin{proof}
As both states $S, S'$ have exactly $t$ already arrived bidders, they both occur only in row $T(t,\cdot)$, or in other words: They both can only be the mechanism's state right after step $t$. 
Note that due to $m\leq k$, all items are assigned in any partial optimum $M_t$ computed after the $k$-th step. Also, remember that $M_{t}$ is unique, and therefore each mechanism step defines a function on the tuple of known bidders and leftover items, together with the choice of the last bidder. 
We show Lemma \ref{lem:independency} via induction over the steps of the mechanism. Let us first have a look at the case $t\leq k$.
In this phase of the mechanism, no items are assigned yet, and therefore the only possible set of leftover items is $J$. 
This reduces the lemma's statement to all sets of $t$ bidders arriving first in equally many $\pi_l$, which is true due to definition.
Now, we assume that $t \geq k+1$, so in the previous step of the mechanism, there might have been an item assigned.\\
For this phase of the mechanism, let us analyze how one row $T(t-1,\cdot)$ results in the subsequent row $T(t,\cdot)$.
We distinguish two cases and analyze for each of them the last step the mechanism made to reach some state $S=(I_t, J_s)\in \mathcal{S}(t,s)$:
\begin{enumerate}
\item In step $t$, there has been no item assigned to the arriving bidder $i_t$. Then, we have that all possible predecessors come from the set
\begin{align*} P_1'(S) =& \left\{ S_i \in \mathcal{S}(t-1,s)\, |\, \condtwocol{\right. \\ & \left.} \exists i\in I_t:\, S_i = (I_t\setminus\{i\}, J_s)   \right \} \end{align*}
By assumption, no item was assigned in step $t$, so only those bidders $i\in I_t$ define a predecessor $S_i$ that get no assignment in the matching $M_t=M_{OPT}(S)$. Let $P_1(S)$ be those $S_i \in P_1'(S)$ for which $i$  fulfills this property. Then, $|P_1(S)|= t-s $ is the number of possible predecessors for this case (as in each $M_t$, all items available are also assigned).
\item In step $t$, some item $j^*$ has been assigned to the arriving bidder. Then, all predecessors must be from the set
\begin{align*}
P_2(S) =& \left\{ S_j\in \mathcal{S}(t-1, s+1)\, |\, \exists j\in J\setminus J_s:\right .\\
&\qquad\left . \, S_j = (I_t\setminus \{i_{M_{OPT}(S)}(j)\}, J_s\cup \{j\})   \right\}
\end{align*}
Here, $i_{M_{OPT}(S)}(j)$ denotes the bidder $i$ which is assigned item $j$ in the matching $M_t= M_{OPT}(S)$.
There are $m-s$ items possible for the role of $j$, and the choice of $j$ also fixes that of the bidder that arrives in step $t$: It has to be the bidder $i_{M_{OPT}(S)}(j)$ who is assigned $j^*$ in the matching $M_{OPT}(S)$.
 Therefore, it holds that $|P_2(S)| = m-s$ for the set of predecessors in this case.
\end{enumerate}
There are no other ways to reach state $S$, so we have $P(S) = P_1(S)\cup P_2(S)$.
Note that all predecessors in the first case are in the same class $\mathcal{S}(t-1,s)$ and all predecessors in the second case are in the same class $\mathcal{S}(t-1,s+1)$.
Now, for any of the above cases, each entry of a predecessor in $T$ leads to an entry of $S$ in $T$ if and only if the next bidder to arrive after this column's $\pi_l^{t-1}$ is exactly the bidder $i\in I_t$ who is not also in $S_i$ (or $S_j$, respectively)\footnote{By Def. of $T$, each $\pi^{t-1}$ leading to a state $S_i$ or $S_j$ occurs in $T$ for every possible choice of $\pi$ that contains $\pi^{t-1}$.}. 
This happens for exactly one of the $n-(t-1)$ possible ways to extend $\pi_l^{t-1}$ towards some $\pi_l$, which is, for a $\frac{1}{n-(t-1)}$-fraction of $mul(S_i)$ (or $mul(S_j)$).
The situation for a predecessor $S_i$ is depicted in 
Table\nobreakspace \ref {table:TransitionToNextStep}.

\begin{table*}
\begin{tabu} to \textwidth {V{3}c|cV{3}cX[c]cX[c]cV{3}}
\hlineB{3}
\multicolumn{2}{V{3}cV{3}}{}  & \multicolumn{5}{cV{3}}{random order}\\
\cline{3-7}
\multicolumn{2}{V{3}cV{3}}{} & $\pi_1 \cdots$ & $\pi_{l_1}$& $\cdots$ & $\pi_{l_2}$ & $\cdots\pi_{n!}$\\
\hlineB{3}
& 1 & & & & & \\
&$\vdots$ & & & & & \\
& $t-1$ & & \ubox{$S_i$} & & \ubox{$S_i$} & \\
step & & &$\big\Downarrow i^*= i$ & &$\big\Downarrow i^*\neq i$ & \\
& $t$ & & \ubox{$S$} & & \ubox{$S' \neq S$} & \\
&$\vdots$ & & & & & \\
& $n$ & & & & & \\
\hlineB{3}
\end{tabu}
\caption{Transition from step $t-1$ to step $t$.}
\label{table:TransitionToNextStep}
\end{table*}

Also, by our induction's assumption, $mul(S_i)$ is the same for all choices of $i\in I_t$, and $mul(S_j)$ is equal for any $j\in J\setminus \{J_s\}$. 
In conclusion, the multiplicity of state $S \in \mathcal{S}(t,s)$ is as follows:
\begin{align}
\label{eq:mul}
mul(S) =& \left(   \frac{t-s}{n-(t-1)}\cdot mul(S_i)\right)\condtwocol{\\ & }\left(  \frac{m-s}{n-(t-1)} \cdot mul(S_j) \right)
\end{align}
This multiplicity holds for any choice of $S\in \mathcal{S}(t,s)$, proving Lemma \ref{lem:independency}.

\end{proof}
At first, it seems surprising that no matter the known bidders, each item is assigned with the same probability. One might think that a reasonable algorithm should increase the probability to assign an item early in the cases where a very high bid on it is presented. However, this is not true: The numerical value of any bid can be arbitrarily high - still, the best any algorithm can ever know about it stays in the form of \emph{\squote{this is the best bid for the item among a random sample of size $t$ out of $n$}}. Loosely spoken, no matter how high any bid on $j$ is in the first half of bidders, there is still a $1/2$ chance a better bid will be presented in the second half. This fact is a consequence of the arrival model itself: So as long as weights are unrestricted, we cannot settle for the good bid in the first half too often, no matter how great it might look.

As we will see, the proven independency of the two sets is a more than handy fact when one is trying to get a bound on the competitive ratio. Note that for the Independency Lemma \ref{lem:independency} to work, our assumptions from above were indeed necessary: Only with uniqueness of the partial optima, we could define how many appearances of any predecessor do lead to an entry of our fixed state $S$. Also, if there were more items in our problem than bidders are sampled ($m>k$), or the graph had too few edges, it would be possible that the same items remain unassigned for most choices of $\pi$ in the first steps after step $k$ (if all bidders prefer more or less the same ones), which would also be in contradiction to the lemma's statement.
Justified by the Independency Lemma \ref{lem:independency}, we can denote the multiplicity of any state in the same class $\mathcal{S}(t,s)$ as $mul(t,s)$ and express the total number of entries in $T$ with $t$ bidders and $s$ items as follows:
\begin{Definition} 
\label{defn:Nrs}
Define $N_{t,s}=mul(t,s)\cdot |\mathcal{S}(t,s)|$.
\end{Definition}

\section{Implications and Proof Outline}
The above independency property guarantees the following:
No matter which bidders have arrived before, each item $j\in J$ is still available with the same, fixed probability when bidder $i_t$ presents himself. Now, if we can show this fixed probability to be sufficiently large, this directly implies that any bidder will \squote{usually} still be offered the \emph{right} item (which he should get in an overall offline optimum) at the time he arrives -- which hopefully, he will then be able to buy. Following this line of thought, the first part of our proof of the competitive ratio is to ensure a certain \textbf{availability} of the items.
Indeed, we will prove the mechanism does sell out slowly enough to make this happen, leaving us with one other problem to address: It is impossible to know during the mechanism which item \emph{is} the \emph{right} one for the current bidder.
We must therefore settle for something reasonably close and, as second part of the proof, ensure the actual \textbf{weight of the chosen edges} is high enough. Recall that we chose to make assignments during the mechanism according to a \emph{partial} offline optimum on the available items and all arrived bidders. This will turn out to be a good choice because of a simple and folklore fact:

\begin{lemma}
\label{partopt}
For $I'\subseteq I$ and $J'\subseteq J$ with $|I'| = c_1\cdot n$ and $|J'|=c_2\cdot m$ chosen uniformly at random from $I$ and $J$, it holds that 
\[\mathbb{E}[OPT(I',J')]\geq c_1\cdot c_2 \, OPT(I,J)\]
\end{lemma}
\begin{proof}
We bound $OPT(I',J')$ from below via a subset of the edges in a fixed, overall optimal matching $M_{OPT}$.
For any edge $(i,j)\in M_{OPT}$, the probability that $i\in I'$ holds is exactly $c_1$, and the probability that $j\in J'$ is exactly $c_2$.
In conclusion, the probability for edge $(i,j)$ to be in the induced subgraph defined by $ (I',J')$ is $c_1 \cdot c_2$.
This holds for any $(i,j)\in M_{OPT}$, proving Lemma \ref{partopt}.
\end{proof}
The Independency Lemma \ref{lem:independency} guarantees that indeed, in each state $S=(I_t, J_s)$ during the mechanism, both sets can be considered chosen uniformly at random.
So as a last step, we will put together the facts that all items are still available with the same, \emph{good} probability and that the weight of the edges we \emph{choose} is reasonably close to that of the edges we \emph{should} choose, yielding a competitive ratio of $e$.
This almost stunning methodical simplicity of the proof is a direct result to our observation of the independency property.

\section{Proof of the Main Theorem}
\subsection{Expected Availability}
With this outline in mind, let us move towards actually proving Theorem~\ref{mainthm}.
We begin in the first part with showing that the mechanism does not sell out too fast. To capture this fact, we strive to determine the expected number of items still available after step $t$, $\mathbb{E}[s_t]$. 

$\mathbb{E}[s_t]$ can be expressed as the sum over the probabilities that each single item is still unassigned, i.e.
\[\mathbb{E}[s_t] = \sum_{j\in J} \Pro{j \text{ is unassigned after step }t}\enspace .\]
This probability, in turn, is the same for all items due to the Independency Lemma \ref{lem:independency}.
We prove the following statement:
\begin{lemma}
For any $j\in J$, $k<t\leq n$:
\[ \Pro{j \text{ is unassigned after step }t}= \frac kt\enspace .\]
\end{lemma}
\begin{proof}

Fix a set of $t$ arrived bidders $I_t$, and a set $J_{t-1}$ of available items after step $t-1$. We analyze the probability for any item $j\in J_{t-1}$ to be assigned in step $t$. 

As $(I_t, J_{t-1})$ fixes the offline partial optimum $M_t$ the mechanism computes in step $t$, it also fixes the bidder $i_{M_t}(j)$ that $j$ is tentatively assigned to. Therefore, the probability that $j$ is assigned in step $t$ by the mechanism is the probability that exactly bidder $i_{M_t}(j)$ is the $t$-th to arrive.

Let $I_{t-1}\subset I_t$ with $|I_{t-1}|+1=|I_t|$. 
As we know from the Independency 
Lemma\nobreakspace \ref {lem:independency},
each possible $S_{t-1}= (I_{t-1}, J_{t-1})$ is the mechanism's state after step $t-1$ with the same probability, or multiplicity in $T$. Now, it follows (because $\pi$ is chosen uniformly at random): In exactly $1/t$ of the runs (columns of $T$) which result in $(I_t, J_{t-1})$ as soon as the $t$-th bidder arrives, this bidder is $i_{M_t}(j)$.
Therefore, the probability for fixed $I_t$ and $J_{t-1}$ that item $j$ is assigned in step $t$ is also $1/t$.

Extending this argument to all possible choices of $(I_t, J_{t-1})$ yields that the overall probability of $j$ being assigned in step $t$ is $1/t$, which holds for all $j$ and all $t>k$.

In consequence,
\[\Pro{j \text{ is unassigned after step }t}=  \prod_{l=k+1}^t \left (1-\frac 1l\right )\]
which we can rewrite as
\[\Pro{j \text{ is unassigned after step }t}= \frac{ \prod_{l=2}^t \left (1-\frac 1l\right )} {\prod_{l=2}^k \left (1-\frac 1l\right )} \enspace.  \]
 With $\prod_{l=2}^t \left(1-\frac 1t\right )= 1/t$ for all $t\in \NN$, $t\geq 2$, we get
\begin{align*}
\Pro{j \text{ is unassigned after step }t}=& \frac{ \prod_{l=2}^t \left (1-\frac 1l\right )} {\prod_{l=2}^k \left (1-\frac 1l\right )}\condtwocol{\\}=\condtwocol{&} \frac{\frac 1t}{\frac 1k}=\frac kt \enspace.
\end{align*}

\end{proof}

Together with the Independency Lemma \ref {lem:independency}, we can conclude
\[\mathbb{E}[s_t]= \sum_{j\in J} \Pro{j \text{ is unassigned after step }t} =m\frac{k}{t}\enspace.\]

\subsection{Weight of Chosen Edges}
In this second part, we need to bound from below the expected weight of the edges our mechanism adds to the matching $M$. We start with the analysis of one single step of the mechanism, fix the number of available items and show the following lemma:

\begin{lemma}
Let $w_{t+1, \pi}$ denote the weight of the edge chosen by the mechanism in step $t+1$ for arrival order $\pi$, if any, and $0$ if no edge is chosen. Define $\Pi_{t,s}$ as the set of arrival orders $\pi$ for which exactly $s$ items are unassigned after step $t$ of the mechanism. Then,
\[ \mathbb{E}[w_{t+1,\pi}|\pi\in\Pi_{t,s}]\geq \frac {1}{n} \frac {s}{m} OPT\]

\end{lemma}
\begin{proof}
The assumption that $s$ items are unassigned is equivalent to the current state of the mechanism being some $S\in \mathcal{S}(t,s)$. 
We consider the arrival orders $\pi$, or columns of $T$, for which this is the case.
By our Independency Lemma \ref {lem:independency}, we know that each such $S$ appears in equally many columns.
In consequence, over all possible choices of $\pi$, after arrival of bidder $i_{t+1}$ at the beginning of step $t+1$ we have a set of $t+1$ bidders and a set of $s$ items which are chosen uniformly at random, and independently from each other. 
In this situation, we can apply the aforementioned folklore guarantee 
(Lemma\nobreakspace \ref {partopt})
to the optimal matching $M_{t+1}$ on $(I_{t+1}, J_s)$:
\[ \mathbb{E}[w(M_{t+1})]\geq  \frac{t+1}{n}\frac{s}{m} OPT\]
As the choice of bidder $i_{t+1}$ is also random, the expected contribution of this bidder to $w(M_{t+1})$ (and therefore to our mechanism) is
\begin{align*}
	\mathbb{E}[w_{t+1,\pi}|\pi\in\Pi_{t,s}] &\geq\frac{1}{t+1} \mathbb{E}[w(M_{t+1})]\\
	&\geq \frac{1}{t+1}  \frac{t+1}{n}\frac{s}{m} OPT\\
	& = \frac{1}{n} \frac{s}{m} OPT
\end{align*}
\end{proof}
This states the expected contribution $\mathbb{E}[w_{t+1,\pi}|\pi\in\Pi_{t,s}]$ of the $(t+1)$-th step under the assumption that the number of available items is $s$. Let us proceed and determine the overall expected contribution $\mathbb{E}[w_{t+1,\pi}]$ of that step. For this, we sum up over all possible values for $s_t$, denoting the size of the set of available items after $t$ steps. Recall that $N_{t,s}$ denotes the total number of entries in $T$ with $t$ bidders and $s$ items (Definition~\ref{defn:Nrs}), i.e. $N_{t,s}=|\Pi_{t,s}|$. We state the minimum number of items available after step $t$ as $s_t^{min}=\max\{0, (m-t+k)\}$. Then it follows via definition of expectation
\begin{align*}
	\mathbb{E}[w_{t+1,\pi}] & = \frac{\sum_{s_t= s_t^{min}}^{m} \mathbb{E}[w_{t+1,\pi}|\pi\in\Pi_{t,s_t}]N_{t,s_t}}{n!}  \\
	& \geq \frac{\sum_{s_t= s_t^{min}}^{m} \frac{1}{n}\frac{s_t}{m}OPT\cdot N_{t,s_t}}{n!} \\
	& = \frac{1}{mn}OPT \frac{\sum_{s_t= s_t^{min}}^{m} s_t\cdot N_{t,s_t}}{n!}\\
	& = \frac 1{mn} OPT\cdot \mathbb{E}[s_t] \enspace .
\end{align*}

\subsection{Proof of Competitive Ratio}
As outlined before, we can now derive a bound on the overall expected welfare of the mechanism: We have to put together what we know about availability of items and contribution of the single assignments.
Let us start with summing these contributions up over all steps after the sampling phase. Then, we get for the welfare of the matching $M$ computed by our mechanism:
\[\mathbb{E}[w(M)] = \sum_{t=k+1}^{n} \mathbb{E}[w_{t,\pi}] \geq \sum_{t=k+1}^{n}\frac{\mathbb{E}[s_{t-1}]}{mn}OPT\]

We have shown $\mathbb{E}[s_t]=m\frac{k}{t}$ for all $t\geq k+1$, which we replace in above inequality.
As we have 
\begin{align*}
	\mathbb{E}[w(M)]  &\geq  \sum_{t=k+1}^{n} \frac{1}{n}\frac{\frac{mk}{t-1}}{m}OPT\\
	& 	= \sum_{t=k+1}^{n} \frac{k}{n}\frac{1}{t-1}OPT\\
	& 	= \frac{k}{n} OPT\sum_{t=k+1}^{n} \frac{1}{t-1}
\end{align*}
the further analysis boils down to a clever choice of $k$, yielding a maximum-possible fraction of the optimal welfare.
What follows is typical to the secretary literature, and can similarly be found e.g. in \cite{KesselheimRTV13}.
We choose $k=\lfloor \frac ne\rfloor$, leading to 
\begin{align*}
\mathbb{E}[w(M)]  \geq&  \frac{\lfloor \frac ne\rfloor}{n} OPT\sum_{t=\lfloor \frac ne\rfloor+1}^{n} \frac{1}{t-1}\\
 =&  \frac{\lfloor \frac ne\rfloor}{n} OPT\sum_{t=\lfloor \frac ne\rfloor}^{n-1} \frac{1}{t}  \enspace .
\end{align*}
It holds that $\sum_{t=\lfloor\frac ne\rfloor}^{n-1}\frac 1t\geq \ln (\frac{n}{\lfloor  n/e\rfloor})\geq 1$, and we get
\[ \mathbb{E}[w(M)]  \geq  \frac{\lfloor \frac ne\rfloor}{n}OPT \enspace .\]
Finally, because of $\frac{\lfloor n/e\rfloor}{n}\geq \frac 1e - \frac 1n$, the mechanism is $e$-competitive (for $n\rightarrow \infty$). This concludes the proof of Theorem \ref{mainthm}.
As even for the simple original version of the secretary problem, this ratio is best-possible, it follows that the mechanism is also optimal.

\section{Discussion and Outlook}

For our mechanism, yielding an optimal competitive ratio of $e$, the analysis relies on only two main properties: A certain, expected availability of the items and the fact that whenever an assignment does happen, the expected value generated can be bounded from below. Both works because of the Independency Lemma \ref {lem:independency}, stating that for the probability with which each item set is still available, it does not matter which bidders the mechanism has already seen.
This sounds surprising, if not unlikely, but turns out to make a lot of sense in hindsight.
Loosely spoken, the mechanism assigns item $j\in J$ if and only if the arriving bidder is \emph{the one} in the first $t$ that is assigned $j$ in an according partial optimum. The chance for this is $1/t$ without any regard to the identity of $j$ \emph{or} the arrived set of bidders. 
The independency in this sense is a direct reflection of the specific type of statement that can be deduced during any algorithm in the secretary model: On the random fraction of data we have seen, this value/assignment is the best to choose.

Considering this, the most remarkable fact lies somewhere else: That the Independency Lemma \ref {lem:independency} holds although our mechanism strays from general design principles in a potentially fatal sense.
Usually, the quality of decisions based on a known, random fraction of the data is best when these are based on \emph{all} the known data, and previous work has done just that.
Still, we choose to disregard part of what we know - weights of all the edges incident to items already assigned - and lose nothing in our guarantee. This is the circumstance inducing our need for a new proof technique in the first place, and the fact that it works is only due to the randomness in the remaining, considered part of the information. 

Based on these general insights as well as the mechanism itself, we give some initial thoughts on a few possible future directions.

\paragraph{Restricted Weights and Multiplicities}
With our initial assumption of adversarial edge weights, it suggests itself to also analyze the impact of possible restrictions on the value or structure of edge weights.
For example, the case of multiplicities $B>1$ for all items comes to mind: Assume that each item in $J$ can be sold $B$ times instead of only once. Here, one can interpret the multiplicity as $B$ individual items, and duplicate edges accordingly. In this surrounding, the mechanism could be combined with the principle of linear progression (which e.g. allows assignment of only at most $t/n \cdot B$ copies of each item up to time $t$), suggesting that it is possible to exploit certain kinds of extra knowledge given on the graph. 

\paragraph{Secretary Problems with Structured Feasible Sets}
As our problem is closely related to the secretary problem on transversal matroids, it is of course an interesting question if there are other cases of secretary problems with strong structural restrictions on their feasible sets that allow for a method in this spirit (special matroids or others). E.g., when arriving agents are elements of a partition matroid (where a set is independent or feasible if and only if it contains at most one element from each set of the partition), our mechanism can be employed to find an independent set of near-optimal weight: Sets of the underlying partition play the role of items, and are marked as unavailable as soon as an element from the set has been selected. Still, the problems our mechanism handles successfully have very specific, strict properties. Finding other such structures that are interesting, but allow for as much control over the mechanism's steps as the one-on-one assignment of the matching case could therefore prove difficult. A small example on difficulties caused by a loss of control is also provided in the next paragraph.

\paragraph{Combinatorial Auctions}
A generalization towards combinatorial scenarios is also a natural idea to follow. However: In general, as soon as assignments are no longer one-on-one, our method runs into additional problems.
When, as the easiest example, a bidder is assigned two items rather than just one in some step, the Independency Lemma \ref {lem:independency} will usually no longer hold afterwards. Even if no item was assigned more likely than any other, certain \emph{combinations} of them might be disappearing with higher probability, and therefore, be available less likely than others in subsequent steps. In other words, it is unrealistic to expect each (size $s$) subset $J'$ of all items to equally likely be available after step $t$. For example, if all bidders demand the set $\{j_1, j_2\}$, our mechanism will likely assign both these items at the same time and therefore, each $J'\subseteq J$ with $|J\cap \{j_1, j_2\}|=1$ would clearly become very unlikely.

Still, this observation does not destroy all hope for such cases. Much more, it indicates a necessity to relax the independency property as stated in this paper, and try to find a more general version hereof: To preserve the general line of the proof, it would suffice to only \emph{bound} the expected availability of items independently from the set of known bidders.

\section{Conclusion}
We have introduced a truthful mechanism for the online weighted bipartite matching problem in the secretary model with competitive ratio of $e$. This closes the gap between the previously known logarithmic-competitive mechanism and the optimal online algorithm without truthfulness. It also implies that even for problems with multi-valued private information and rich combinatorial structure, the secretary model allows for near-optimal, truthful mechanisms. This settles a question which has been open for a number of years. 

Our proof technique is based on an observation of the independency between the set of known bidders and assigned items. 
This independency reduces the proof of the competitive ratio to a few, elementary calculations. The reason for this is that it actually reflects on the mechanism's behavior as being very reasonable, making decisions only to an extent that is justified by the nature of available information due to the random arrival model.
We expect the idea to transfer nicely to certain other problems with one-on-one-assignments, including of course any special cases to our problem which might be of interest. For problems in the scope of combinatorial auctions, e.g. the case of submodular valuations, it is clear that our method (more exactly, the proof of the Independency Lemma \ref {lem:independency}) cannot directly be applied. However, it appears as if a weaker version hereof should also suffice, and we hope to identify such a suitable variation at a later point. 

Beyond this new sense of hope for the design of truthful mechanisms in the secretary model, we want to stress that the design principles reflected in this paper are purely algorithmic and could therefore prove useful even for problems where no good algorithm is known yet.

\section{Acknowledgements}
The author would like to thank Thomas Kesselheim for his valuable thoughts and insights.

\appendix

\end{document}